\documentclass[12pt, reqno]{amsart}
\usepackage{amssymb, mathrsfs}
\usepackage{latexsym}
\usepackage{amsmath, commath}
\usepackage{amsthm}
\usepackage{amsfonts}
\usepackage{dsfont}
\usepackage{mathtools}
\usepackage{epsfig}
\usepackage{verbatim}
\usepackage{graphicx}
\usepackage{accents}
\usepackage{mathdots}
\usepackage{stmaryrd}

\usepackage{pstricks-add}

\usepackage{tikz}
\usetikzlibrary{cd}
\usetikzlibrary{patterns}
\usetikzlibrary{calc}
\usepackage{pgfplots}
\pgfplotsset{compat=1.11}

\usepackage[left=1.4in,right=1.4in]{geometry}
\usepackage{enumerate}
\usepackage[colorlinks=true,linkcolor=blue,citecolor=blue]{hyperref}

\DeclareMathOperator{\GL}{GL}
\DeclareMathOperator{\Tr}{Tr}

\newtheorem*{lemma*}{Lemma}

\newtheorem{theorem}{Theorem}[section]
\newtheorem{proposition}[theorem]{Proposition}
\newtheorem{lemma}[theorem]{Lemma}

\newtheorem{corollary}[theorem]{Corollary}

\newtheorem*{open question}{Open Question}

\newtheorem*{question*}{Question}

\theoremstyle{definition}

\newtheorem*{claim*}{Claim}
\newtheorem{definition}[theorem]{Definition}
\newtheorem{example}[theorem]{Example}

\theoremstyle{remark}
\newtheorem{remark}[theorem]{Remark}

\usepackage{xstring}

\newcommand{\wamatrix}[1]{%
  \StrRemoveBraces{#1}[\waouterraw]%
  \StrSubstitute{\waouterraw}{,}{\\}[\waouter]%
  \StrRemoveBraces{\waouter}[\wainnerraw]%
  \StrSubstitute{\wainnerraw}{,}{&}[\wainner]%
  \begin{pmatrix}
    \wainner
  \end{pmatrix}%
}

\numberwithin{equation}{section}

\begin{document}
\title[Explicit Matrices With Complexity $4n-\mathrm{o}(n)$ and Local Logic Gates]{Explicit Matrices over $\mathbb Z_2$ with CNOT and Row Complexity $4n-\mathrm{o}(n)$ and Local Logic Gates}
\author{Sherry Gong}
\address[]{SG: Texas A\&M University, College Station, TX 77840 USA}
\email{sgongli@tamu.edu}
\author{Andrew Yu}
\address[]{AY: Harvard University, Cambridge, MA 02138 USA}
\email{andrewyu45@gmail.com }

\thanks{Supported in part by NSF and ARO MURI grant W911NF-20-1-0082.}

\begin{abstract}

In this article, we present an explicit family of invertible $n\times n$ matrices over $\mathbb Z_2$ whose CNOT and row complexity is at least $4n-\text{o}(n)$; equivalently, reducing these matrices to the identity requires at least $4n-\text{o}(n)$ elementary row operations. Moreover, the same complexity lower bound holds in the stronger computational model where the CNOT gates are replaced by arbitrary local linear logic gates, namely arbitrary invertible linear transformations acting on pairs of coordinates.


Let $G_n$ denote the permutation group generated by local logic gates acting on the set of binary strings of length $n$. We prove that $G_n$ is naturally isomorphic to the group of all invertible affine transformations of the vector space $\mathbb Z_2^n$, thus reducing the problem of estimating the quantum complexity of permutations in $G_n$ to the row reduction complexity of invertible matrices over $\mathbb Z_2$. As an application, we show that the permutations associated with our explicit matrices have quantum complexity at least $4n-\text{o}(n)$.

\end{abstract}

\maketitle

\section{Introduction}

$$\\$$

Constructing explicit invertible matrices over the field $\mathbb Z_2$ with large CNOT and row-reduction complexity is a longstanding open problem. Given an invertible matrix over $\mathbb Z_2$, the finite field with two elements, its \emph{row complexity} is the minimum number of elementary row operations required to reduce it to the identity matrix. Its \emph{CNOT complexity} is the minimum number of Type III elementary row operations required for the same reduction, where a Type III operation consists of adding one row to another.
While a simple counting argument shows that almost all invertible matrices require at least $cn^2/\log n$ elementary row operations, proving nontrivial lower bounds for explicit matrices has proved to be remarkably difficult. Gowers highlighted this problem by asking for an explicit matrix requiring a large linear  number of elementary row operations \cite{Gowers2000, Gowers2009}.

The asymptotic theory of CNOT complexity has been well studied. A classical counting  A classical counting argument shows that almost all matrices have CNOT complexity $\Omega(n^2/\log(n))$ \cite{PatelMarkovHayes2008}.

Despite these counting arguments, the construction of explicit families of matrices that have superlinear CNOT length is a major open question in the field. Indeed, until very recently, the most prominent bounds in the CNOT metric have been in the $3n$ range \cite{Bataille2022,BuFanJoo2025,ChristensenJorgensenPavlogiannisVandePol2025}. Much of the work to this effect has centered around permutation matrices with long cycles, which have been shown to need circuits of length $3(n-1)$ \cite{BuFanJoo2025}.

In this paper, we construct an explicit family of invertible $n\times n$ matrices over $\mathbb Z_2$ whose CNOT and row complexity is at least $
4n-\mathrm{o}(n).$

While this note was in preparation, J\o rgensen independently posted an explicit $4m-o(m)$ construction also based on Sergeev's result in \cite{Jorgensen2026}. Despite our result being based on a similar matrix, we believe that the main step in our proof is substantially different, in that we directly showed that a matrix of the form $\wamatrix{{{U,I},{I,U^T}}}$ was invertible instead of embedding it into a larger invertible matrix as J\o rgensen's asymptotic construction did.

Our work is motivated in part by the complexity problem of local logic gates in quantum computation, which was proposed to us by Michael Freedman. A qubit is represented by a unit vector, and a quantum computation is a unitary transformation implemented as a sequence of  logic gates.
The quantum complexity of a unitary transformation is the minimum number of logic gates required to implement it (cf. \cite{Brown2021, BrownFreedmanLinSusskind2021, BuGarciaJaffeKohLi2022, FreedmanKitaevLarsenWang2003, Kitaev1997, Nielsen2006, NielsenChuang, NielsenDowlingGuDoherty2006, Shor1994, Lin2019, Aaronson2016}). In this article, we study permutations generated by local logic gates acting on binary strings of length $n$. We prove that the permutation group generated by these gates is naturally isomorphic to the group of all invertible affine transformations of the vector space $\mathbb Z_2^n$. Moreover, every elementary matrix is realized by a local logic gate, and every local logic gate is an affine transformation whose linear part is either the identity, an elementary matrix, or a product of two elementary matrices, and whose translation part is supported on at most two coordinates.
This characterization reduces the problem of estimating the quantum complexity of permutations generated by local logic gates to the row-reduction complexity of invertible matrices over $\mathbb Z_2$. As applications, we obtain a subquadratic upper bound for the quantum complexity of all such permutations, prove that almost all of them have a matching asymptotic lower bound, and show that the permutations corresponding to our explicit matrices have quantum complexity at least $4n-\mathrm{o}(n)$.

The article is organized as follows. In Section~2, we construct an explicit family of invertible matrices over $\mathbb Z_2$ with CNOT and  row complexity at least $4n-\mathrm{o}(n)$. In Section~3,  we describe the concept of local logic gates introduced to us by Freedman. We identify the permutation group generated by local logic gates with the group of invertible affine transformations over $\mathbb Z_2$ and derive a quantum complexity upper bound. In Section~4, we prove that almost all permutations generated by local logic gates have quadratic-over-logarithmic quantum complexity. Finally, in Section~5, we relate the quantum complexity of affine transformations to the row complexity of their linear parts and deduce the quantum complexity lower bound $4n-\mathrm{o}(n)$ for the permutations arising from our explicit matrices.

\section{ 
An explicit  invertible $n\times n$  matrix with CNOT and row complexity lower bound $4n-\text{o}(n)$.   }

$$\\$$

Let $n$ be a positive integer. Let  $A\in GL(n, \Bbb{Z}_2)$ be an $n\times n$ invertible matrix over $\Bbb{Z}_2 $, the finite field of two elements.  The CNOT length $\ell_n(A)$  of $A$ is defined to be the minimal number of type III elementary matrices (the type where one row gets added to another) whose product is $A$. We call the Type III elementary matrices \emph{transvections}. The row complexity of $A$, $\ell_{row}(A)$,  is defined to 
be the minimal number of elementary matrices  whose product is $A$. 

 In this section, we construct explicit $n\times n$ matrices with CNOT length at least $4n-\text{o}(n)$.

One may wonder if our lower bounds are an artifact of the set of gates chosen for our model of linear reversible CNOT circuits. However, we find similar bounds for a stronger computational  model where we permit more gates. 

Let $\lambda_n(A)$ denote the 2-local gate length of $A$, that is, the number of gates needed to implement $A$ as a reversible circuit when one is allowed to apply an arbitrary element of $\GL(2,\mathbb{Z}_2)$ to an arbitrary pair of coordinates (instead of only being allowed to add one coordinate to another as in the CNOT length). Since every CNOT gate and every elementary matrix is a $2$-local linear gate, we have
\[
\lambda_n(A)\le \ell_{\mathrm{row}}(A)\le \ell_n(A).
\]
We also show that our explicit $n\times n$ matrices have 2-local gate length at least  $4n-\text{o}(n)$, which suggests that our lower bound for $\ell_n$ is not a peculiarity of the choice of CNOT gates.

The proof of our result rests on a recent breakthrough of Sergeev, in which he showed a $5n-o(n)$ bound for the XOR circuit complexity of an explicit family of $n \times n$ matrices \cite{Sergeev2025}. Because additive circuits are not required to be reversible, the main difficulty of our proof was modifying the matrices to be in $\GL(n, \mathbb{Z}_2)$ while not losing too much of the circuit complexity.

Our result also implies that   our  explicit family  of $n\times n$ matrices have row complexity  at least $4n-o(n)$. As shown in Section~5, the corresponding permutations generated by local logic gates also have quantum complexity at least $4n-o(n)$.

\subsection{Definitions and background}

For $1 \leq i < j \leq n$ and $B \in GL(2, \Bbb{Z}_2)$ let $p_{ij}(B) \in GL(n, \Bbb{Z}_2)$ be obtained by applying $B$ to the $i$ and $j$ coordinates among $x_1, x_2, \ldots, x_n$. For example, $p_{ij}\left(\wamatrix{{{1,1},{0,1}}}\right)$ is given an elementary row operation where row $j$ is added to row $i$.

Let $\lambda_n(A)$ denote the minimal number of matrices of the form $p_{ij}(B) \in GL(n, \Bbb{Z}_2)$ for $B \in GL(2, \Bbb{Z}_2)$ whose product is $A$.  
It is clear that $\lambda_n(A) \leq \ell_n(A)$. 
As discussed in Section 4, most matrices in $GL(n, \Bbb{Z}_2)$ have $\lambda_n$ and $\ell_n$ both $\geq \frac{n^2}{2\log(n)}$.
However explicit constructions of matrices with $\lambda_n(A)$ and $\ell_n(A)$ greater than $3n$ has been elusive. In this section, we construct explicit matrices with $\lambda_n(A)$ and $\ell_n(A)$ greater than $4n-\text{o}(n)$.  


We will make use of results on additive circuits. For an $r\times n$ matrix
$A$ over $\mathbb{Z}_2$, let $C_\oplus(A)$ denote the additive
(XOR) circuit complexity of $A$. Following Sergeev, we shall also denote
this quantity by $L(A)$.
 This is the number of XOR operations (additions in $\Bbb{Z}_2$) needed in a program that computes $Ax$ for $x = \wamatrix{{{x_1},{x_2},{\vdots},{x_n}}}$. In more mathematical terms, the additive complexity $L(A)$ can be defined as follows.

\begin{definition}[Additive circuit]
Let
\[
X=(x_1,\dots,x_n)^T
\]
be a collection of variables over $\mathbb Z_2$.

An \emph{additive circuit} is a sequence of linear forms
\[
y_1,\dots,y_{n+L}
\]
defined by
\[
y_1=x_1,\quad \dots,\quad y_n=x_n,
\]
and
\[
y_{n+t}=y_i+y_j
\]
for some indices
\[
i,j<n+t,
\qquad
t=1,\dots,L.
\]

The integer $L$ is called the \emph{complexity} of the circuit. The variables
\[
y_{o_1},\dots,y_{o_m}
\]
are designated as outputs.
\end{definition}

\begin{definition}[Additive complexity]
Let
\[
A\in M_{m\times n}(\mathbb Z_2),
\]
and let \(a_1,\dots,a_m\) denote the rows of \(A\).

The additive complexity $L(A)$ is the minimum complexity of an additive
circuit whose outputs are
\[
\langle a_1,X\rangle,\dots,\langle a_m,X\rangle.
\]
\end{definition}

In other words, the \emph{additive complexity} $L(A)$ is the smallest integer $L$ for which there exist vectors
\[
v_1,\dots,v_{n+L}\in \mathbb Z_2^n
\]
satisfying the following conditions:

\begin{enumerate}
\item
We have
\[
v_1=e_1,\dots,v_n=e_n,
\]
where $e_1,\dots,e_n$ are the standard basis vectors of $\mathbb Z_2^n$;

\item for each \(t=1,\dots,L\),
\[
v_{n+t}=v_i+v_j
\]
for some indices \(i,j<n+t\);

\item every row $a_1,\dots,a_r$ of $A$ occurs among the vectors
\[
v_1,\dots,v_{n+L}.
\]
\end{enumerate}

Equivalently, $L(A)$ is the minimum number of vector additions required to construct all rows of $A$ starting from the standard basis vectors of $\mathbb Z_2^n$. It is clear $L(A)\leq l_{row}(A).$

Following Sergeev's notation in \cite{Sergeev2025}, for two (not necessarily square) matrices $A$ and $B$ denote $A \boxplus B$ to be $\wamatrix{{{A,0},{0,B}}}$. 

The following is an immediate consequence of Sergeev's paper \cite{Sergeev2025} and the transposition principle:

\begin{theorem}[Sergeev]  If $B$ is an $a \times b$ matrix with $a > b$ such that $B$ does not have any rows with exactly one $1$, and there is $k \geq 3$ such that any $2k$ rows from $B$ are linearly independent, then 
\[L(B \boxplus B^T) \geq 3a+ 2\frac{2k-4}{2k-1}a^{1-1/k}-3b.\]
In particular, if $k \gg log(a)$, then $L(B \boxplus B^T) \geq 5a-O\left(\frac{a\log a}{k}\right)-3b.$
\label{Sergeev_theorem}
\end{theorem}

\begin{remark}
The theorem above is an immediate consequence of Theorem~2 of Sergeev \cite{Sergeev2025}
together with the transposition principle.

Following Sergeev \cite{Sergeev2025}, let
\[
A\in M_{m\times n}(\mathbb Z_2)
\]
with rows $a_1,\dots,a_m$.

Consider an extended additive circuit computing the linear operator
\[
X\mapsto AX,
\]
where
\[
X=(x_1,\dots,x_n)^T
\]
are the essential input variables and
\[
Y=(y_1,\dots,y_t)^T
\]
are additional nonessential input variables.

Every gate computes a linear form
\[
\langle a,X\rangle+\langle b,Y\rangle,
\]
where
\[
a\in\mathbb Z_2^n,\qquad b\in\mathbb Z_2^t.
\]

The vector $b$ is called the type of the gate. Its weight is the Hamming
weight
$
\operatorname{wt}(b),
$ the number of non-zero coordinates in the vector $b$.

The reduced complexity of the circuit is defined as
\[
\#\{\text{addition gates in the circuit}\}
-
\#\{\text{distinct type vectors } b
\text{ with }\operatorname{wt}(b) \geq 2\}.
\]

The reduced additive complexity $L^*(A)$ is the minimum reduced complexity
among all extended additive circuits computing $AX$.

Sergeev proved that if $B$ satisfies the hypotheses of the theorem, then
\[
L^*(B)
\ge
a+\frac{2k-4}{2k-1}a^{1-1/k}-b.
\]

The transposition principle states that for every matrix
$B\in M_{a\times b}(\mathbb Z_2)$,
\[
L(B^T)\ge L^*(B)+a-b.
\]

Since
\[
L(B\boxplus B^T)
\ge
L^*(B)+L(B^T),
\]
we obtain
\[
L(B\boxplus B^T)
\ge
2L^*(B)+a-b.
\]
Substituting Sergeev's lower bound for $L^*(B)$ yields
\[
L(B\boxplus B^T)
\ge
3a+
2\frac{2k-4}{2k-1}a^{1-1/k}
-3b,
\]
which is exactly the statement of Theorem \ref{Sergeev_theorem}.
\end{remark}

In particular, Sergeev gave the following example.

\begin{example}(Sergeev \cite{Sergeev2025})
Let $p$, $s$, $N$ be positive integers with $N>sp$ and $s>1$. Let $f(x)\in \mathbb{Z}_2[x]$ be an irreducible
polynomial of degree \(p\). Note that such a polynomial has been constructed deterministically in polynomial time by Shoup in \cite{Shoup1990Irreducible}

We define the finite field with $2^p$ elements by
\[
\mathbb{F}_{2^p}
=
\mathbb{Z}_2[x]/(f(x)).
\]

The elements of $\mathbb{F}_{2^p}$ are residue classes of polynomials modulo
$f(x)$. Since $f(x)$ is irreducible, the ideal $(f(x))$ is maximal, and
therefore $\mathbb{Z}_2[x]/(f(x))$ is a field.

Every element of $\mathbb{F}_{2^p}$ can be represented uniquely in the form
\[
a_0+a_1\alpha+\cdots+a_{p-1}\alpha^{p-1},
\qquad
a_i\in\mathbb{Z}_2,
\]
where $\alpha$ denotes the image of $x$ in the quotient
$\mathbb{Z}_2[x]/(f(x))$. Consequently,
\[
|\mathbb{F}_{2^p}|=2^p.
\]

 Let $\alpha_1, \alpha_2, \cdots \alpha_m$ be pairwise distinct non-zero elements in $\Bbb{F}_q = \Bbb{F}_{2^p}$. $\Bbb{F}_q$ is a vector space over $\Bbb{Z}_2$  of dimension $p$, so let us consider a fixed basis of $\Bbb{F}_{2^p}$ over $\Bbb{Z}_2$.  Then for the matrix 
\[U = \begin{pmatrix}
\alpha_1 & \alpha_1^2 & \cdots & \alpha_1^s \\
\alpha_2 & \alpha_2^2 & \cdots & \alpha_2^s \\
\vdots & \vdots & \ddots & \vdots \\
\alpha_N & \alpha_N^2 & \cdots & \alpha_N^s
\end{pmatrix} \]
which is an $N \times s$ matrix over $\Bbb{F}_q$, we may view it as an $N \times sp$ matrix over $\Bbb{Z}_2$ using our basis of $\Bbb{F}_q$ over $\Bbb{Z}_2$. Then Sergeev observed that this matrix satisfies all the conditions required of $B$ in  Theorem \ref{Sergeev_theorem}  with $k = \lfloor s/2 \rfloor$. (The linear independence of the subsets of rows is from the fact that you can take size $s$ subsets of the rows and see it as a Vandermonde matrix over $\Bbb{F}_q$, and the fact that no row has exactly one 1 comes from the fact that since $\alpha_i \neq 0$, each of $\alpha_i, \alpha_i^2, \ldots \alpha_i^s$ will contribute at least one $1$ to the row over $\Bbb{Z}_2$, so the row sums are at least $s$.)

Sergeev then plugged $s = \lceil \sqrt{n} \rceil$, $p =\lceil \log(n) \rceil$ and $N = n$ to get his matrix $B$ with dimensions $ \approx n \times \sqrt{n}$ and $A$ with dimensions $\approx (n + \sqrt{n} )\times ( n + \sqrt{n})$ and $L(A) \geq 5n-\text{o}(n).$ 

It is easy to check that if $s$ were chosen to be $\lceil n^{1/3} \rceil$ instead of $\lceil \sqrt{n} \rceil$, the resulting matrix $U \boxplus U^T$ would still satisfy $L(U \boxplus U^T) \geq 5n-o(n)$, because we would have $k \asymp n^{1/3} \gg \log n$ and $sp \asymp n^{1/3} \log n = \text{o}(n)$

\end{example}

\subsection{Main result and its proofs}

In order to use these results we will show the following lemma that relates $\lambda_n$ to $L$.

\begin{lemma} Let $A \in GL(n, \Bbb{Z}_2)$ and let $P$ be a permutation matrix, i.e. an invertible matrix with exactly one $1$ in every column and every row. Then $L(A+P) \leq \lambda_n(A) +n$.
\label{comparison_lemma}
\end{lemma}

\begin{proof}
Let $\lambda = \lambda_n(A)$ and let $A = g_\lambda g_{\lambda-1} \ldots g_1$. We would like to show that $(A+P)x$ can be calculated from $x$ with $\lambda+n$ XOR gates. Since each nonzero term $P$ contributes to $(A+P)x$ can be calculated with one operation (by adding the relevant original input variable to the relevant part of $Ax$), it suffices to show that $Ax$ can be calculated with $\lambda$ operations. 

Note that each $g_i$ has the effect of taking two terms $(u,v)$ out of the original $x$ and replacing them with one of the following (or its reverse order): $(u,v)$, $(u, u+v)$, or $(v, u+v)$. Note that reversing the order of two terms in $x$ is cost-free when it comes to additive circuit complexity $L$, and that each of these three operations can be computed with at most 1 XOR gate. Thus $L(A) \leq \lambda$. 

Note that each row of $P$ contains exactly one 1, so once the outputs $Ax$ can be computed, we can add the corresponding input variable to each output using just one XOR gate. Thus, $L(A+P) \leq \lambda + n = \lambda_n(A) + n$

\end{proof}






Now we attempt to apply Sergeev's example to construct an invertible matrix in our situation. Set $q=2^p, N=q-1, s=  \lceil N^{1/3}\rceil,$  and $m=N+ps$. 

We would like to use the same $N \times sp$ matrix $U$ and consider a matrix of the form $$\wamatrix{{{U,I_N},{I_{sp},U^T}}} = U \boxplus U^T + \wamatrix{{{0, I_N},{I_{sp},0}}},$$ where the matrix $U$ is the one from Sergeev's example. The only subtlety is that we must choose $U$ so that our matrix $\wamatrix{{{U,I_N},{I_{sp},U^T}}}$ is invertible.

Note that $\wamatrix{{{U,I_N},{I_{sp},U^T}}}$ is invertible if and only if there is no nonzero $\wamatrix{{{x},{y}}}$ for $x \in (\Bbb{Z}_2)^{ps}$ and $y \in (\Bbb{Z}_2)^{N}$ such that $\wamatrix{{{U,I_N},{I_{sp},U^T}}} \wamatrix{{{x},{y}}} = 0$. But these equations can be re-written as 
\begin{align*}
    Ux + y &= 0 \\
    x + U^T y &= 0.
\end{align*}
But this is equivalent to
\begin{align*}
    y &= Ux \\
    x + U^T Ux &= 0.
\end{align*}
So to show that there is no non-zero solution it suffices to show that $U^TU = 0$. We now show this holds for large $n$ for $s = \lceil n^{1/3} \rceil$. 

\begin{lemma}
For $U$ taken as above (as in Sergeev's example), for sufficiently large $p$,
we have that $U^TU = 0$. 
\label{UT_U_zero}
\end{lemma}

\begin{proof} For $i = 1, 2, \cdots, p$,   let $\phi_i:\Bbb{F}_{2^p} \to \Bbb{Z}_2$ denote the coordinate maps for the $p$ basis vectors of $\Bbb{F}_{2^p}$ over $\Bbb{Z}_2$. Then the rows of $U^T$ (which are transposes of the columns of $U$) have the form 
$$(\phi_i(\alpha_1^t), \phi_i(\alpha_2^t), \ldots \phi_i(\alpha_N^t)  ),$$
for $i \in \{1, 2, \ldots, p\}$ and $t \in \{1, 2, \ldots s\}$ (there are $ps$ rows in $U^T$). 

Thus, we wish to show that for any $i,j \in \{1, 2, \ldots, p\}$ and $t,u \in \{1, 2, \ldots s\}$, 
\begin{equation}
    \sum_{\alpha \in \Bbb{F}_{2^p}^\times} \phi_i(\alpha^t) \phi_j(\alpha^u)
    \label{dot_product_in_terms_of_phi_i}
\end{equation}
is zero.

Recall that for finite field extensions of finite fields, every linear map from the extension to the base field can be written in terms of the field trace. In particular, in our case we can write 
$\phi_i(x) = \Tr(\beta x)$ and $\phi_j (x) = \Tr(\gamma x)$ for $\beta, \gamma \in \Bbb{F}_{2^p}^\times$

Thus the left hand side of \eqref{dot_product_in_terms_of_phi_i} can be rewritten as
\begin{equation}
    \sum_{\alpha \in \Bbb{F}_{2^p}^\times} \Tr(\beta\alpha^t) \Tr(\gamma\alpha^u)
    \label{dot_product_in_terms_of_Tr}
\end{equation}

Since finite extensions of finite fields are Galois extensions, the field trace can be written as the sum of Galois conjugates of $x$, which in this case is 
$$\Tr(x) = x + x^2 + x^4 + \cdots x^{2^{p-1}}.$$

Thus we can rewrite \eqref{dot_product_in_terms_of_Tr} as
\begin{align*}
    \sum_{\alpha \in \Bbb{F}_{2^p}^\times} \sum_{i,j \in \{0, 1, \ldots p-1\}} (\beta\alpha^t)^{2^i} (\gamma\alpha^u)^{2^j} & =\sum_{i,j \in \{0, 1, \ldots p-1\}}  \sum_{\alpha \in \Bbb{F}_{2^p}^\times} (\beta\alpha^t)^{2^i} (\gamma\alpha^u)^{2^j} \\
    & = \sum_{i,j \in \{0, 1, \ldots p-1\}}  \beta^{2^i}\gamma^{2^j}\sum_{\alpha \in \Bbb{F}_{2^p}^\times} \alpha^{2^it+2^ju}
    \label{dot_product_in_terms_of_Tr}
\end{align*}

It now suffices to show that for $i, j \in \{0, 1, \ldots , p-1\}$ and $t, u \leq \lceil N^{1/3} \rceil$,
\[\sum_{\alpha \in \Bbb{F}_{2^p}^\times} \alpha^{2^it+2^ju} = 0.\]
Since $\Bbb{F}_{2^p}^\times$ is cyclic of order $2^p-1$, we have that $$\sum_{\alpha \in \Bbb{F}_{2^p}^\times} \alpha^{2^it+2^ju} = 0$$ 
unless $q-1=2^p-1 | 2^it+2^ju$. 

Note that if $i = j$, then $2^p-1 | 2^it+2^ju$ if and only if $2^p-1| t+u$, but $0<t+u < 2^p-1$ for sufficiently large $p$, so this is impossible. We may thus assume without loss of generality that $i<j$. Then we divide into two cases:

\noindent\textbf{Case 1: $j-i\leq p/2$:}
Note that $2^p-1|2^it+2^ju$ if and only if $2^p-1|t+2^{j-i}u$. However, for large $p$, $t,u < (2^p-1)^{0.35}$ and $2^{j-i}<(2^p-1)^{0.55}$, so $t+2^{j-i}u<(2^p-1)$, and thus $2^p-1$ cannot divide $t+2^{j-i}u$.

\noindent\textbf{Case 2: $j-i > p/2$:}
Note that $2^p-1|2^it+2^ju$ if and only if $2^p-1|2^{p-j+i}t+2^{p}u$, but $2^{p-j+i}t+2^{p}u \equiv 2^{p-j+i}t+u$ and now $p-j+i \leq p/2$, so we have reduced this to the previous case.

This completes the proof of \eqref{dot_product_in_terms_of_phi_i}, and thus that $U^TU = 0$.
\end{proof}

\begin{theorem}
With $N = 2^p-1$, $s = \lceil N^{1/3} \rceil$, $m = N+ps$, and $U$ as above, the matrix $A = \wamatrix{{{U, I_N},{I_{ps}, U^T}}}$ is invertible for sufficiently large $p$, i.e. is an element in $\GL(m, \mathbb{Z}_2)$ and has $\lambda_m(A) \geq 4m-\text{o}(m)$ and $\ell_m(A) \geq 4m-\text{o}(m)$.
\end{theorem}
\begin{proof}
By Lemma \ref{UT_U_zero} we have that $A$ is invertible. Let $P = \wamatrix{{{0, I_N}, {I_{ps}, 0}}}$, which is a permutation matrix. Then $A+P = \wamatrix{{{U, 0},{0, U^T}}} = U \boxplus U^T$. 

As we explained above, by Theorem \ref{Sergeev_theorem}, $L(A+P) = L(U \boxplus U^T) \geq 5N-\text{o}(N)$. Thus, by Lemma \ref{comparison_lemma}, 
\[\lambda_m(A) \geq L(A+P) - m \geq 5N-\text{o}(N) -(N+ps) = 4N-\text{o}(N).\]
Since $ps = \text{o}(N)$, this is $4m-\text{o}(m)$. Since $\lambda_m(A) \leq \ell_m(A)$, the same lower bound holds for $\ell_m(A)$. 
\end{proof}



\section{A complexity upper bound for permutations generated by local logic gates }\label{sec:pre}

$$\\$$

In this section, we show that the permutation group generated by local logic gates on  the set of binary strings of length $n$
can be identified with the group of invertible affine transformations of the vector space $  \Bbb{Z}_2 ^n   $ over the field
 $  \Bbb{Z}_2$. We also prove that each local logic gate is an affine transformation whose linear part is either the identity, an elementary matrix, or a product of two elementary matrices, and whose translation part is supported on at most two coordinates.
  As a consequence, we obtain $\frac{2n^2}{\mbox{log}_2 n -1} + 2n \mbox{ log}_2 n+ \frac{3n}{2} +1$ as a complexity upper bound 
 for all permutations generated by local logic gates on the set of binary strings of length $n$.

We identify the set of binary strings of length $n$ with the vector space
$  \Bbb{Z}_2 ^n   $ over the field $  \Bbb{Z}_2$. 
We write a vector $v$ in $  \Bbb{Z}_2 ^n   $ in the following column form:

 $$  v =\begin{pmatrix}
c_1\\
c_2\\
\vdots \\
c_n
\end{pmatrix},   $$ 
where $c_i \in   \Bbb{Z}_2.$

An affine transformation $T: \Bbb{Z}_2 ^n   \rightarrow \Bbb{Z}_2 ^n$
can be written as:

$$Tv =Av + v_0$$
for all $v\in  \Bbb{Z}_2 ^n ,$
where $A$ is a linear transformation from $\Bbb{Z}_2 ^n $ to $\Bbb{Z}_2 ^n $
and $v_0$ is a fixed vector in $ \Bbb{Z}_2 ^n $.

An affine transformation is the composition of a linear transformation $A$ with a translation by a vector $v_0$. The linear transformation $A$ is called the linear part of the affine transformation $T$.
An affine  transformation $T$ is a bijection if and only if the linear transformation is invertible. When an affine transformation $T$ is a bijection, we say that $T$ is invertible.

$$\\$$

We first consider the case of $n=2$. The general case is a consequence of the $2$-dimensional case.

\begin{lemma} When $n=2$, the group of all invertible affine transformations is equal to the group of all permutations on
the set of binary strings of length $2.$
\label{lemma3.1}
\end{lemma}

\noindent{Proof of Lemma \ref{lemma3.1}}:
First, we observe that the group of all invertible affine transformations is a subgroup of the group of all permutations on the set of binary strings of length $2$.
Since the group of all permutations on the set of binary strings of length $2$
 has $4!=24$ elements, it suffices to prove that the group of all invertible affine transformations has $24$ elements.

The following six $2$ by $2$ matrices is a complete list of invertible matrices over the field $\Bbb{Z}_2$:

\[
  \left(\begin{array}{cc}
    1&0\\
    0&1
  \end{array}\right), ~~~
  \left(\begin{array}{cc}
    0&1\\
    1&0
  \end{array}\right), ~~~
  \left(\begin{array}{cc}
    1&1\\
    0&1
  \end{array}\right),~~~
  \left(\begin{array}{cc}
    1&0\\
    1&1
  \end{array}\right),~~~
  \left(\begin{array}{cc}
   0&1\\
   1&1
  \end{array}\right),~~~
  \left(\begin{array}{cc}
  1&1\\
  1&0
  \end{array}\right).
\]

Notice that we have four choices of the vector $v_0$ in the definition of the affine transformation $T$. As such, there are a total of $24$ invertible affine transformations. \qed

$$\\$$

We remark that the six matrices in the proof of the above lemma
will be used to construct linear logic gates. Notice that the two matrices in the middle are the CNOT gates. The CNOT gates generate the other four matrices
as follows:

\[ \left(\begin{array}{cc}
    1&0\\
    0&1
  \end{array}\right) 
  = \left(\begin{array}{cc}
    0&1\\
    1&0
  \end{array}\right)^2,
 \]
 
 \[ \left(\begin{array}{cc}
    0&1\\
    1&0
  \end{array}\right)
  =  \left(\begin{array}{cc}
    1&0\\
    1&1
  \end{array}\right)
  \left(\begin{array}{cc}
    1&1\\
    0&1
  \end{array}\right)
   \left(\begin{array}{cc}
    1&0\\
    1&1
  \end{array}\right),\]
 \[ \left(\begin{array}{cc}
    0&1\\
    1&1
  \end{array}\right)
  =\left(\begin{array}{cc}
    1&0\\
    1&1
  \end{array}\right)
  \left(\begin{array}{cc}
    0&1\\
    1&0
  \end{array}\right).
  \]
  
$$\\$$

We now define the concept of  local logic gates. 

Let $\sigma $ be a permutation on the set of binary strings of length $2,$
 identified as the vector space  $  \Bbb{Z}_2 ^2  $.
For any $1\leq i<j\leq n$, $\sigma $
induces a permutation $\sigma_{i,j}$ on the set of binary strings of length $n$,
  the vector space $  \Bbb{Z}_2 ^n  $, as follows:

$$\sigma_{i,j}( c_1, \cdots, c_{i-1}, c_i, c_{i+1}, \cdots, c_{j-1}, c_j,  c_{j+1},\cdots, c_n) $$ $$=
(c_1, \cdots, c_{i-1}, p_1(\sigma (c_i, c_j)), c_{i+1},\cdots, c_{j-1}, p_2 (\sigma (c_i, c_j)), c_{j+1}, \cdots, c_n),$$ 

where $p_1 (a, b)=a$ and $p_2(a, b)=b$ for all $(a, b)\in  \Bbb{Z}_2 ^2  $.  

$$\\$$

The following concept of local logic gate was introduced to the authors by Professor Michael Freedman.

\begin{definition} For any permutation $\sigma$ of the set of binary strings of length 2 and any pair $1\leq i<j\leq n$, the induced permutation $\sigma_{i,j}$ on the set of binary strings of length $n$
is called a local logic gate.
\label{definition3.2}
\end{definition}

$$\\$$

The following theorem is a consequence of Lemma \ref{lemma3.1}.

\begin{theorem}

Assume that $n\ge 2$. The permutation group $G_n$ generated by all local logic gates on the set of binary strings of length $n$ is equal to the group of all invertible affine transformations on the vector space $\Bbb{Z}_2 ^n $ over the field $\Bbb{Z}_2$.
Furthermore, every local logic gate acting on coordinates $i$ and $j$ has the form
\[
\begin{pmatrix}x_i\\x_j\end{pmatrix}
\longmapsto
B\begin{pmatrix}x_i\\x_j\end{pmatrix}+b,
\]
where $B\in GL(2,\mathbb Z_2)$ and $b\in\mathbb Z_2^2$. Its linear part is
either the identity, an elementary matrix, or a product of two elementary
matrices, and its translation part is a translation  on at most two coordinates.

\label{theorem3.3}
\end{theorem}

We remark that the linear logic gates generate the group of invertible linear transformations.

$$\\$$

\noindent{Proof of Theorem  \ref{theorem3.3}}:
Let $\operatorname{AGL}(n,2)$ be the group of invertible affine transformations. 
By Lemma \ref{lemma3.1}, every permutation of $\mathbb Z_2^2$ is an invertible affine
transformation. Therefore, every local logic gate on $\mathbb Z_2^n$ is an
invertible affine transformation supported on two coordinates. Hence:
\[
G_n\subseteq \operatorname{AGL}(n,2).
\]

Conversely, let
\[
T(v)=Av+v_0
\]
be an invertible affine transformation of $\mathbb Z_2^n$.

Since $A\in GL(n,\mathbb Z_2)$, it can be written as a product of row swaps
and row additions. A row swap on coordinates $i$ and $j$ is induced by the
permutation
\[
(a,b)\longmapsto (b,a)
\]
of $\mathbb Z_2^2$, and a row addition is induced by one of the permutations
\[
(a,b)\longmapsto (a+b,b),
\qquad
(a,b)\longmapsto (a,a+b).
\]
Thus every row swap and every row addition is a local logic gate.

It remains to consider the translation by $v_0$. Write
\[
v_0=w_1+\cdots+w_r,
\]
where each $w_\nu\in\mathbb Z_2^n$ is  a vector supported on at most two coordinates.
Then
\[
\tau_{v_0}
=
\tau_{w_1}\cdots \tau_{w_r}.
\]
Each $\tau_{w_\nu}$ is induced by a translation of $\mathbb Z_2^2$, and hence
is a local logic gate.

Hence both the linear map $A$ and the translation $\tau_{v_0}$ belong to
$G_n$. Since
$
T=\tau_{v_0} A,
$
we obtain
$
\operatorname{AGL}(n,2)\subseteq G_n.
$
As a consequence, we have
$
G_n=\operatorname{AGL}(n,2).
$

Finally, every local logic gate acting on coordinates $i$ and $j$ has the form
\[
\begin{pmatrix}x_i\\x_j\end{pmatrix}
\longmapsto
B\begin{pmatrix}x_i\\x_j\end{pmatrix}+b,
\]
where $B\in GL(2,\mathbb Z_2)$ and $b\in\mathbb Z_2^2$.
Among the six elements of $GL(2,\mathbb Z_2)$, one is the identity, three are
elementary matrices, and the remaining two are products of two elementary
matrices. The translation vector $b$ is supported on at most two coordinates.
\qed
 
$$\\$$

\begin{theorem} All permutations in the group $G_n$ generated by local logic gates acting on the set of binary strings of length $n$ have complexity at most
\[
\frac{2n^2}{\log_2 n -1}
+2n\log_2 n
+n
+\left\lceil\frac{n}{2}\right\rceil,
\]
where the complexity of a permutation is the minimum number of local logic gates whose composition is equal to the permutation.
\label{theorem3.4}
\end{theorem}

$$\\$$

\noindent{Proof of Theorem \ref{theorem3.4}:} By Section 3 of [13], all invertible $n$ by $n$ matrices over the field $ \Bbb{Z}_2$ is a product of at most
$\frac{2n^2}{\mbox{log}_2 n -1} + 2n \mbox{ log}_2 n+ n$  elementary matrices.
Observe that any translation is a product of at most $ \lceil \frac{n}{2} \rceil $ local logic gates. Theorem \ref{theorem3.4} now follows from Theorem \ref{theorem3.3}. \qed

$$\\$$

\section{A superlinear lower bound for the complexity of almost all permutations in $G_n$}\label{sec:pre}

$$\\$$

In this section, we prove that almost all permutations in $G_n$ have complexity at least
\[
q_n=
\left\lfloor
\frac{n^2+n-2\log_2 n}
{\log_2(12n(n-1))}
\right\rfloor
\]
with respect to the local logic gates.

$$\\$$

\begin{theorem}
The group $G_n$ has $2^n (2^n-1) (2^n-2)\cdots (2^n-2^{n-1})$ elements.
\label{theorem4.1}
\end{theorem}

$$\\$$

\noindent{Proof of Theorem \ref{theorem4.1}}:
Observe that there are $2^n$ choices for the vector $v_0$ in the vector space $\Bbb{Z}_2 ^n$ in the definition
of an affine transformation on the vector space $\Bbb{Z}_2^n$  in Section 2. Consequently, Theorem \ref{theorem4.1} follows from the fact that the group of all invertible linear
transformations has $(2^n-1) (2^n-2)\cdots (2^n-2^{n-1})$ elements.
This can be seen as follows. Let $A$ be an invertible linear transformation over the field  $\Bbb{Z}_2$.
The first row of $A$ has $2^n-1$ choices since we need to rule out the choice of $0$ as the first row. The second row has $2^n-2$ choices since we need to rule out the choices of $0$ and the first row. The third row has $2^n-2^2$ choices since we need to rule out the choices of all possible linear combinations of the first two rows.
Continuing like this, the last row has $2^n-2^{n-1}$ choices since we need to rule out the choices of all possible linear combinations of the first $n-1$ rows
(each such linear combination corresponds to a choice of any subset of the set of the first $n-1$ rows). 
\qed

$$\\$$

\begin{theorem}
Let
\[
q_n=
\left\lfloor
\frac{n^2+n-2\log_2 n}
{\log_2(12n(n-1))}
\right\rfloor .
\]
Then, as $n\to\infty$, almost all permutations in $G_n$ have complexity
 greater than $q_n$ with respect to the local logic gates.
\label{theorem4.2}
\end{theorem}

\noindent{Proof of Theorem \ref{theorem4.2}}:
There are
$
24\binom{n}{2}=12n(n-1)
$
choices of a local logic gate together with a pair of coordinates on which it
acts. Let
$
N_n=12n(n-1).$
Let $P_q$ denote the set of elements of $G_n$ that can be implemented by at
most $q$ local logic gates.

Since the identity permutation of $\mathbb Z_2^2$ induces the identity map on
$\mathbb Z_2^n$, the identity map is itself a local logic gate. Hence every
element of $P_q$ can be expressed as a product of exactly $q$ local logic
gates by inserting identity gates if necessary. Therefore the multiplication
map
\[
S^q\longrightarrow P_q,
\qquad
(g_1,\ldots,g_q)\mapsto g_1\cdots g_q,
\]
is surjective, where $S$ denotes the set of all local logic gates. Since
$|S|=24\binom{n}{2}=12n(n-1)$, we obtain
\[
|P_q|
\le
|S|^q
=
\bigl(12n(n-1)\bigr)^q.
\]

Hence
\[
|P_{q_n}|\le N_n^{q_n}.
\]

By the definition of $q_n$,
\[
q_n\log_2 N_n
\le n^2+n-2\log_2 n.
\]
As a consequence, we have
\[
|P_{q_n}|
\le
2^{\,n^2+n-2\log_2 n}
=
\frac{2^{n^2+n}}{n^2}.
\]

On the other hand, Theorem \ref{theorem4.1} gives
\[
|G_n|
=
2^n\prod_{j=0}^{n-1}(2^n-2^j).
\]
Factoring $2^n$ from each term in the product, we obtain
\[
\begin{aligned}
|G_n|
&=
2^n\prod_{j=0}^{n-1}
2^n\bigl(1-2^{j-n}\bigr)\\
&=
2^{n^2+n}
\prod_{j=0}^{n-1}\bigl(1-2^{j-n}\bigr)\\
&=
2^{n^2+n}
\prod_{r=1}^{n}(1-2^{-r}).
\end{aligned}
\]

The infinite product
$
\prod_{r=1}^{\infty}(1-2^{-r})
$
converges to a positive number. Hence there exists a constant $c>0$,
independent of $n$, such that
\[
\prod_{r=1}^{n}(1-2^{-r})\ge c
\]
for all $n$. As a consequence, we have
$
|G_n|\ge c\,2^{n^2+n}.$

It follows that
\[
\frac{|P_{q_n}|}{|G_n|}
\le
\frac{2^{n^2+n}/n^2}
{c\,2^{n^2+n}}
=
\frac{1}{cn^2}.
\]
Since
$
\frac{1}{cn^2}\longrightarrow 0,
$
the proportion of elements of $G_n$ that can be implemented using at most
$q_n$ local logic gates tends to zero.

Thus almost all permutations in $G_n$ have complexity  greater than
$q_n$. \qed
$$\\$$

\section{Quantum Complexity of  Affine Transformations}\label{}

$$\\$$

In this section, we estimate the quantum complexity of affine transformations
in terms of the quantum complexity of their linear parts and  translation parts.
As a consequence, we show that the complexity problem for permutations generated by local logic gates is essentially equivalent to the row reduction complexity problem of invertible matrices over the field $\Bbb{Z}_2$.
$$\\$$

\begin{definition}

(1) Let $T$ be an invertible affine transformation of $\mathbb Z_2^n$. The
\emph{quantum complexity} of $T$, denoted by $QC(T)$, is the minimum number
of local logic gates whose composition is equal to $T$.

(2) Let $A$ be an invertible linear transformation of $\mathbb Z_2^n$. The
\emph{linear quantum complexity} of $A$, denoted by $QC_L(A)$, is the minimum
number of linear local logic gates whose composition is equal to $A$.
\end{definition}

$$\\$$

\begin{lemma} If $A$ is an invertible linear transformation on the vector space $\Bbb{Z}_2^n$ over the field $\Bbb{Z}_2$,  then we have
$$QC_L(A)=QC(A).$$

\label{lemma5.2}

\end{lemma}

$$\\$$

\noindent{Proof of Lemma \ref{lemma5.2}}: It is clear that $QC(A)\leq QC_L(A).$
Hence it suffices to show that $QC_L(A) \leq QC(A).$ 

Claim: If $T_0$ and $T_1$ are affine transformations defined by:
$T_0 v=A_0 v+ v_0$ and $T_1 v=A_1 v+ v_1$  with $A_0$ and $A_1$, respectively, as the linear parts of $T_0$ and $T_1$, we have 
$$(T_0T_1)_L= A_0A_1, $$ where $(T_0T_1)_L$ is the linear part of the affine transformation $T_0T_1$.

This claim follows from the following identity:

$$ (T_0T_1)v = A_0(A_1v+v_1)+v_0= A_0A_1v+(A_0v_1+v_0).$$

Let $l=QC(A).$ There exist local logic gates $\rho_1, \cdots, \rho_l$ such that
$$A=\rho_1\cdots\rho_l.$$
Let $(\rho_i)_L$ be the linear part of $\rho_i$. The assumption that $\rho_i$ is a local logic gate implies that $(\rho_i)_L$ is a linear local logic gate. 
By the above claim, we have 
$$A=(\rho_1)_L\cdots (\rho_l)_L.$$
It follows that $$QC_L(A) \leq l = QC(A).$$ \qed

By Lemma \ref{lemma5.2}, we know $QC_L(A)=QC(A)=\lambda_n(A),$ where $\lambda_n(A)$ is defined in Section 2. 
As such, the explicit $n\times n$ matrix constructed in Section 2 has a quantum complexity lower bound of $4n-\text{o}(n)$.

$$\\$$

\begin{theorem}
Let $T$ be the affine transformation as follows: $Tv=Av+v_0$ on the vector space $\Bbb{Z}_2 ^n$ with $A$ as its linear part.
We have 
$$ QC_L (A) \leq QC(T)\leq QC_L (A) + QC(\tau_{v_0}),$$
where $\tau_{v_0}$ is the translation by $v_0$.
\label{theorem5.3}
\end{theorem}

$$\\$$

\noindent{Proof of Theorem \ref{theorem5.3}:}
Clearly, we have 
$$QC(T)\leq QC (A) + QC(\tau_{v_0}). $$
Hence by Lemma \ref{lemma5.2}, we obtain
$$QC(T)\leq QC_L (A) + QC(\tau_{v_0}).$$
By the proof of Lemma \ref{lemma5.2}, we also have
$$QC_L(A)=QC(A)\leq QC(T).$$ \qed

$$\\$$

The following theorem states that the complexity of an affine transformation can be estimated by the linear quantum complexity of its linear part.
Hence the complexity of an affine transformation is essentially a problem about the row complexity of its linear part.

\begin{corollary}

Let $T$ be an affine transformation: $Tv=Av+v_0$ on the vector space 
$  \Bbb{Z}_2 ^n   $
with $A$ as its linear part. We have 
$$ QC_L (A) - \left\lceil \frac{n}{2} \right\rceil  \leq  QC(T)\leq QC_L (A) + 
\left\lceil \frac{n}{2} \right\rceil. $$

\label{corollary5.4}

\end{corollary}

Finally, we give an example showing that the inequality can not be improved to an equality.

\begin{example}
Let \(n=2\), and define
\[
        T(x_1,x_2)=(x_1+x_2+1,x_2).
\]
Then \(T(v)=Av+v_0\), where
\[
        A(x_1,x_2)=(x_1+x_2,x_2),
        \qquad
        v_0=(1,0).
\]
We have
\[
        QC_L(A)=1,
        \qquad
        QC(\tau_{v_0})=1.
\]
However, \(T\) itself is a single local logic gate, since it is a permutation
of \(\mathbb Z_2^2\). Hence
\[
        QC(T)=1.
\]
Therefore
\[
        QC(T)<QC_L(A)+QC(\tau_{v_0}).
\]
Thus, the inequality
\[
        QC(T)\leq QC_L(A)+QC(\tau_{v_0})
\]
cannot be improved to an equality in general.
\end{example}

$$\\$$

\noindent{Proof of Corollary \ref{corollary5.4}:}
Our corollary follows from Theorem \ref{theorem5.3} and the inequality 
$$QC(  \tau_{v_0})\leq  \left\lceil \frac{n}{2} \right\rceil.$$
\qed

\begin{definition} 
Let $T$ be an invertible affine transformation of $\mathbb Z_2^n$. An
\emph{elementary affine operation} is an affine transformation supported on
at most two coordinates whose linear part is either the identity, a row swap,
or a row addition.
We define the elementary  quantum complexity  $QC_E(T)$ of $T$ to be the minimum number of elementary affine operations
whose composition is $T$.
\end{definition}

The following result indicates that the two notions of complexity are essentially equivalent.

\begin{proposition}
Let \(T:\Bbb{Z}_2^n\to \Bbb{Z}_2^n\) be an invertible affine
transformation. We have
\[
    QC(T)
    \leq
    QC_E(T)
    \leq
   2 QC(T).
\]
\end{proposition}

\begin{proof}
The first inequality follows because every elementary affine operation is a
local logic gate.  Indeed, row swaps, row additions, and translations in at
most two coordinates are all permutations of \(\mathbb Z_2^2\) acting on a
chosen pair of coordinates.

We prove the second inequality.  A local logic gate is induced by a
permutation of the four-element set
$ \Bbb{Z}_2^2$.
By Lemma \ref{lemma3.1} every such permutation is affine. 
Therefore every local logic gate acting on coordinates \(i,j\) has the form
\[
\begin{pmatrix}x_i\\x_j\end{pmatrix}
\longmapsto
B\begin{pmatrix}x_i\\x_j\end{pmatrix}+b,
\]
where \(B\in GL_2(\Bbb{Z}_2)\) and \(b\in \Bbb{Z}_2^2\).

We now observe that each element of \(GL_2(\Bbb{Z}_2)\) can be written
as product of at most two elementary linear operations of the following types:
\[
        x_i\leftrightarrow x_j,
        \quad
        x_i\leftarrow x_i+x_j,
        \quad
        x_j\leftarrow x_j+x_i.
\]
Indeed, \(GL_2(\mathbb Z_2)\) has six elements:
\[
I,\quad
\begin{pmatrix}0&1\\1&0\end{pmatrix},\quad
\begin{pmatrix}1&1\\0&1\end{pmatrix},\quad
\begin{pmatrix}1&0\\1&1\end{pmatrix},\quad
\begin{pmatrix}0&1\\1&1\end{pmatrix},\quad
\begin{pmatrix}1&1\\1&0\end{pmatrix}.
\]
The first is the identity, the next three are elementary linear operations,
and the last two are products of a row swap and one row addition.  For
example,
\[
\begin{pmatrix}1 &1\\1&0\end{pmatrix}
=
\begin{pmatrix}1&1\\0&1\end{pmatrix}
\begin{pmatrix}0&1\\1&0\end{pmatrix}.
\]

The translation vector \(b\in\Bbb{Z}_2^2\) is a translation in at most two
coordinates, and hence is one of the allowed affine elementary operations.
Moreover, this translation may be combined with one of the two elementary
linear operations above, since our elementary affine operations are allowed
to include translations in the same two coordinates.

Thus every local logic gate can be expressed as a product of at most two
elementary affine operations.  Consequently, if \(T\) is a product of \(k\)
local logic gates, then \(T\) is a product of at most \(2k\) elementary
affine operations.  Taking \(k=QC(T)\), we obtain
\[
QC_E(T)\leq 2QC(T).
\]
Together with the first inequality, this proves the proposition.
\end{proof}

$$\\$$

\noindent {\bf Acknowledgement:} The authors are deeply grateful to Professor Michael Freedman for his generous guidance, insightful suggestions, and for introducing us to the complexity problem of local logic gates. His encouragement and perspective have been invaluable throughout this work.

AY would also like to express sincere gratitude to Professor Arthur Jaffe, Professor Kaifeng Bu, and Dr. Liyuan Chen for their mentorship, encouragement, and many inspiring discussions relating to Sections 3-5 of this project.
$$\\$$


\end{document}